\documentclass[twocolumn,amsmath,amssymb, prl, superscriptaddress]{revtex4-1}
\usepackage{color}
\usepackage{graphicx}
\usepackage{dcolumn}
\usepackage{bm}
\usepackage{mathtools}
\usepackage{verbatim}
\usepackage{hyperref}
\definecolor{darkred}  {rgb}{0.5,0,0}
\definecolor{darkblue} {rgb}{0,0,0.5}
\definecolor{darkgreen}{rgb}{0,0.5,0}
\hypersetup{
  pdftitle = {Schmidt Rank},
  pdfauthor = {},
  colorlinks = true,
  urlcolor  = blue,         
  linkcolor = darkblue,     
  citecolor = darkgreen,    
  filecolor = darkred       
}

\newtheorem{theorem}{\textbf{Theorem}}
\newtheorem{lemma}{\textbf{Lemma}}
\newtheorem{corollary}{\textbf{Corollary}}
\newtheorem{proposition}{\textbf{Proposition}}

\newcommand{\ket}[1]{|#1\rangle}
\newcommand{\bra}[1]{\langle#1|}

\newcommand{\op}[2]{|#1\rangle\langle #2|}
\newcommand{\id}{\text{id}}
\newcommand{\mbb}{\mathbb}
\newcommand{\mc}{\mathcal}
\newcommand{\mbf}{\mathbf}
\newcommand{\tr}{\text{Tr}}
\newcommand{\wh}{\widehat}

\begin{document}

\title{The Zero-error Entanglement Cost is Highly Non-Additive}

\author{Qiuling Yue}
\email{qiuling.yue@siu.edu}
\affiliation{State Key Laboratory of Networking and Switching Technology, Beijing University of Posts and Telecommunications, Beijing 100876, China}
\affiliation{Department of Physics and Astronomy, Southern Illinois University,
Carbondale, Illinois 62901, USA}
\author{Eric Chitambar}
 \email{echitamb@illinois.edu}
\affiliation{Department of Electrical and Computer Engineering, Coordinated Science Laboratory,\\ University of Illinois at Urbana-Champaign, Urbana, IL 61801}
%

\date{\today}

\begin{abstract}
The Schmidt number is an entanglement measure whose logarithm quantifies the zero-error entanglement cost of generating a given quantum state using local operations and classical communication (LOCC).  
In this paper we show that the Schmidt number is highly non-multiplicative in the sense that for any integer $n$, there exists states whose Schmidt number remains constant when taking $n$ copies of the given state.  These states also provide a rare instance in which the regularized zero-error entanglement cost can be computed exactly.  We then explore the question of increasing the Schmidt number by quantum operations.  We describe a class of bipartite quantum operations that preserve the Schmidt number for pure state transformations, and yet they can increase the Schmidt number by an arbitrarily large amount when generating mixed states.  Our results are obtained by making connections to the resource theory of quantum coherence and generalizing the class of dephasing-covariant incoherent operations (DIO) to the bipartite setting.

\end{abstract}


\maketitle


\section{\label{sec:level1} Introduction}

Quantum entanglement and quantum coherence describe two of the most prominent features of quantum systems, and they represent strong departures from the classical description of matter.  Within quantum information science, entanglement has been characterized as an operational resource that can be employed in a ``distant-lab'' setting to perform some task, such as teleportation \cite{Plenio-2007a}.  This notion of resource is made precise in the framework of a quantum resource theory \cite{Horodecki-2013a, Brandao-2015a, Chitambar-2018a}.  Broadly speaking, a quantum resource theory studies quantum information processing under a restricted class of operations, and certain quantum states, called resource states, cannot be generated by these operations.

In the resource theory of entanglement, local operations and classical communication (LOCC) are the allowed operations, and the states that cannot be generated by LOCC, i.e. entangled states, are the resource states.  Other resource theories have been studied in the literature that still hold entanglement (or at least NPT entanglement) as a resource, but they nevertheless allow for operations distinct from LOCC, such as separable operations \cite{Rains-1997a, Vedral-1997a}, non-entangling operations \cite{Harrow-2003a, Brandao-2011a}, and PPT operations \cite{Rains-1999a}.  While these latter resource theories lack the direct physical motivation characteristic of LOCC, they are nevertheless valuable to understanding the structure of quantum entanglement and deriving analytic bounds on what is possible by LOCC.  Another such operational ``cousin'' (or perhaps ``second cousin'') to LOCC will be studied in this paper with the aim of shedding new light on the nature of the zero-error entanglement cost for generating a given quantum state.

Motivated by results in the resource-theoretic study of entanglement, much interest has recently been placed in formulating a resource theory of quantum coherence \cite{Aberg-2006a, Baumgratz-2014a, Yadin-2016a, Winter-2016b}.  The essential aim behind this research direction is to capture coherent superposition as a fungible non-classical resource possessed by certain quantum states that can be quantified and manipulated.  Many similarities exist between the resource theories of coherence and entanglement.  One example of this commonality discussed in this paper is the Schmidt/coherence number and its operational interpretation as the zero-error resource cost.  We begin in Sections \ref{Sect:Schmidt_Number}  and \ref{Sect:Coherence_Number} by reviewing these objects in entanglement and coherence theories, respectively.  We then investigate different operational properties of the coherence number in Section \ref{Sect:DIO} and formally map these results to the Schmidt number of entanglement theory in Section \ref{Sect:MCDC}.

\section{The Schmidt Number and Entanglement Cost of a Bipartite Quantum State}

\label{Sect:Schmidt_Number}

For a bipartite state $\rho^{AB}$, the task of state preparation, or entanglement dilution, involves transforming an $m$-dimensional maximally entangled state $\ket{\Phi_m}=\sqrt{1/m}\sum_{i=0}^{m-1}\ket{i}\ket{i}$ into $\rho$ using LOCC.  In practice, one would like to minimize the amount of entanglement consumed in this task, and this optimal amount quantifies the \textit{zero-error entanglement cost} of the state $\rho$, defined as
\begin{equation}
\label{Eq:zero-error-entanglement-cost}
E^{(0)}_c(\rho)=\min \{\log m:\inf_{\Lambda\in\mc{L}}||\Lambda(\Phi_m)-\rho||_1=0\},
\end{equation}
where $\Phi_m=\op{\Phi_m}{\Phi_m}$, $||A||_1=\tr\sqrt{A^\dagger A}$, and the infimum is taken over the set $\mc{L}$ of all LOCC maps acting on $\Phi_m$.  
In the definition of $E^{(0)}_c(\rho)$, the preparation of the target state $\rho$ is required to be perfect, a condition typically too stringent for experimental implementations.  Instead, it is common to consider a many-copy preparation and/or allow for $\epsilon$-error \cite{Buscemi-2011a}.  Intuitively, one might expect that preparing two copies of a bipartite state $\rho$ in, say, the zero-error case, requires twice as much entanglement as when preparing just one copy.  Surprisingly this intuition turns out to be incorrect, i.e $E_c^{(0)}(\rho^{\otimes 2})\not=2E_c^{(0)}(\rho)$, a phenomenon known as non-additivity.

To understand the origin of non-additivity in zero-error state preparation, it is helpful to use a more mathematically analyzable characterization of $E_c$ than its operational definition in Eq. \eqref{Eq:zero-error-entanglement-cost}.  From basic results in the study of LOCC entanglement transformations, it follows that the zero-error entanglement cost of a state $\rho$ is given precisely by the logarithm of its Schmidt number \cite{Nielsen-1999a, Hayashi-2006a}.  Recall that the Schmidt number of a pure state $\ket{\varphi}$, denoted as $r_S(\varphi)$ and also called the Schmidt rank, is the smallest $m$ such that $\ket{\varphi}$ can be expressed as a sum of $m$ product states, i.e. $\ket{\varphi}=\sum_{i=1}^m\sqrt{\lambda_i}\ket{\alpha_i}\ket{\beta_i}$.  In Ref. \cite{Terhal-2000a},  Terhal and Horodecki extended this definition to mixed states as
\begin{equation}
r_S(\rho)=\inf_{\{p_i,\ket{\psi_i}\}}\max_ir_S(\psi_i),
\end{equation}
where the minimization is taken over all pure state ensembles such that $\rho=\sum_ip_i\op{\psi_i}{\psi_i}$.  

With the equivalence $E_c^{(0)}(\rho)=\log_2 r_S(\rho)$, non-additivity of $E_c^{(0}$ can then be seen as arising from the freedom that exists in representing $\rho^{\otimes 2}$ as a mixture of pure states, a freedom which allows for entanglement across the Hilbert space of each copy.  In their original paper, Terhal and Horodecki already discovered that the Schmidt number is non-additive.  Specifically, they showed that $E^{(0)}_c(\rho^{\otimes 2})=E^{(0)}_c(\rho)$ for certain entanglement Werner states.  A natural subsequent question to consider is how non-additive the Schmidt number can be.  As our first main result, we show that the non-additivity can be arbitrarily large.  More precisely, for any $n\in\mbb{N}$, we describe a two-qubit entangled state $\rho$ such that
\begin{equation}
r_S(\rho^{\otimes n})=r_S(\rho).
\end{equation}
Physically this means that just as much entanglement is needed to prepare $n$ copies of the state by LOCC than is needed to prepare just a single copy.

Given the non-additivity of $E_c^{(0)}$, a perhaps more appropriate way to measure the intrinsic resource cost of exactly preparing a given state is through regularization.  The regularized zero-error entanglement cost is defined as
\begin{align}
E_c^{(0,\infty)}(\rho)&=\lim_{n\to\infty}\frac{1}{n}E_c^{(0)}(\rho^{\otimes n})\notag\\
&=\lim_{n\to\infty}\frac{1}{n}\log_2\left[r_S(\rho^{\otimes n})\right].
\end{align}
This can be equivalently expressed as the optimal asymptotic rate of entangled bits (ebits) $\ket{\Phi_2}$ consumed per copy of $\rho$ perfectly generated.  That is,
\begin{equation}
E^{(0,\infty)}_c(\rho)=\inf\{r :\lim_{n\to\infty}\{\inf_{\Lambda\in\mc{L}}||\Lambda(\Phi^{\otimes \lceil rn\rceil }_2)-\rho^{\otimes n}||_1=0\}\}.
\end{equation}
A more familiar quantity is called the entanglement cost of the state $\rho$ \cite{Bennett-1996a, Hayden-2001a}, and it is defined as
\begin{equation}
E_c(\rho)=\inf\{r :\lim_{n\to\infty}\{\inf_{\Lambda\in\mc{L}}||\Lambda(\Phi^{\otimes \lceil rn\rceil}_2)-\rho^{\otimes n}||_1\}=0\}.
\end{equation}
The key difference between $E_c^{(0,\infty)}$ and $E_c$ is that $E_c$ allows for vanishingly small error whereas the error tolerance in $E_c^{(0,\infty)}$ is strictly zero.  One easily sees that
\begin{equation}
E_c\leq E_c^{(0,\infty)}\leq E_c^{(0)}.
\end{equation}
In general all three of these quantities are difficult to compute.  However, below we determine their values for a family of two-qubit maximally correlated mixed-states.  To our knowledge, this represents the first time that $E_c^{(0,\infty)}$ has been analytically computed for a genuinely mixed state.

To obtain the above results, we study the analogous questions in the resource theory of quantum coherence.  We then formally map what we learn in coherence theory to the case of entanglement.

\section{The Coherence Number and Coherence Cost of a Quantum State}

\label{Sect:Coherence_Number}

Recent approaches to a resource theory of quantum coherence begin by specifying some orthonormal basis $\{\ket{i}\}_{i=0}^{d-1}$ for the state space $\mc{H}$ of some $d$-dimensional quantum system, referred to as the incoherent basis \cite{Baumgratz-2014a}.  Coherence is then defined and quantified with respect to this basis.  The free states are those that are diagonal in the incoherent basis, and any state possessing non-zero off-diagonal terms is deemed a resource state.  For convenience, in what follows we will denote the set of all density matrices acting on Hilbert space $\mc{H}$ by $\mc{D}(\mc{H})$ (or simply $\mc{D}$), the set of all diagonal states by $\mc{I}(\mc{H})$ (or simply $\mc{I}$), and the completely positive trace-preserving (CPTP) map which fully dephases every input state by $\Delta$, i.e. $\Delta(\rho)=\sum_i\op{i}{i}\rho\op{i}{i}$.  Note that $\rho\in\mc{I}$ if and only if $\Delta(\rho)=\rho$.

Different types of allowed quantum operations have been proposed, but all of them forbid the generation of non-diagonal states from diagonal ones.  The most well-studied operational class is called \textit{incoherent operations} (IO), and it consists of CPTP maps $\Lambda$ admitting a Kraus operator representation $\{K_\lambda\}_\lambda$ in which $K_\lambda\rho K_\lambda^\dagger= \Delta(K_\lambda\rho K_\lambda^\dagger)$ for all $\lambda$ and all $\rho\in\mc{I}$ \cite{Baumgratz-2014a}.  A subset of IO is the class of \textit{strictly incoherent operations} (SIO), and an SIO map $\Lambda$ satisfies the further restriction that $K_\lambda \Delta(\rho) K_\lambda^\dagger=\Delta(K_\lambda\rho K_\lambda^\dagger)$ for all $\lambda$ and all $\rho\in\mc{D}$ \cite{Winter-2016b, Yadin-2016a}.  Note that every SIO map $\Lambda$ commutes with $\Delta$, but the converse is not true.  The collection of all \textit{dephasing covariant incoherent operations} (DIO) are those satisfying $\Lambda\circ\Delta=\Delta\circ\Lambda$, and they represent a third class of incoherent operations \cite{Chitambar-2016b, Marvian-2016a}.

Analogous to entanglement, the task of resource cost can then be studied using any of these operations.  For operational class $\mc{O}\in\{\text{SIO,IO,DIO}\}$, the zero-error cost, regularized zero-error cost, and coherence cost are respectively defined as
\begin{align}
C^{(0)}_{c,\mc{O}}(\rho)&=\min \{\log m:\inf_{\Lambda\in\mc{O}}||\Lambda(\phi_m)-\rho||_1=0\};\notag\\
C^{(0,\infty)}_{c,\mc{O}}(\rho)&=\inf\{r :\lim_{n\to\infty}\{\inf_{\Lambda\in\mc{O}}||\Lambda(\Phi^{\otimes \lceil rn\rceil}_2)-\rho^{\otimes n}||_1=0\}\}\notag\\
C_{c,\mc{O}}(\rho)&=\inf\{r :\lim_{n\to\infty}\{\inf_{\Lambda\in\mc{O}}||\Lambda(\Phi^{\otimes \lceil rn\rceil}_2)-\rho^{\otimes n}||_1\}=0\}.
\end{align}
Here $\ket{\phi_m}=\frac{1}{\sqrt{m}}=\sum_{i=0}^{m-1}\ket{i}$ is the $m$-dimensional maximally coherent state \cite{Baumgratz-2014a}.
The asymptotic coherence cost was first proposed by Winter and Yang in Ref. \cite{Winter-2016b} for SIO and IO, and it was later studied for DIO in \cite{Chitambar-2018c}.  The zero-error coherence cost was studied by Zhao \textit{et al.} in \cite{Zhao-2018a} for all three operational classes.  The zero-error coherence costs under SIO and IO are equivalent, and it can be expressed as
\begin{equation}
C_{c,\text{SIO/IO}}^{(0)}(\rho)=\inf_{\{p_i,\ket{\psi_i}\}}\max_i\log_2[r_C(\psi_i)],
\end{equation}
where the minimization is taken over all pure state ensembles for $\rho$, and $r_c(\psi)$ is the number of nonzero terms in the expansion $\ket{\psi}=\sum_{i=1}^dc_i\ket{i}$.  The quantity $r_c(\psi)$ is called the \textit{coherence number}, and it plays the analog of the Schmidt number in coherence theory.

For pure states, the zero-error IO coherence cost is additive since $r_c(\psi^{\otimes n})=r_c(\psi)^n$.  For mixed states, this equivalence breaks down in dramatic fashion, as the following example shows.  Consider the family of qubit states
\begin{align}
\omega_\alpha=\frac{1-\alpha}{2}\mbb{I}+\alpha\phi_2.
\end{align}
Each state in this family represents a mixing of the maximally coherent state with the totally mixed state.
\begin{theorem}
\label{thm:1}
Let $n\in\mbb{N}$ be arbitrary and take $0<\alpha\leq 2^{1/n}-1$.  Then
\begin{equation}
r_c(\omega_\alpha^{\otimes n})=r_c(\omega_\alpha)=2.
\end{equation}
\end{theorem}
\begin{proof}
First observe that $r_c(\omega_{\alpha}^{\otimes n})\geq 2$ for any $\alpha>0$ since $\omega_\alpha$ is an incoherent state if and only if $\alpha=0$.  So it remains to show that this lower bound is tight.  In the following we let $\mbf{i}=(i_0,i_1,\cdots,i_{n-1})$ denote an element in $\{0,1\}^{\times n}$, and similarly $\ket{\mbf{i}}=\ket{i_0}\ket{i_1}\cdots\ket{i_{n-1}}$.  Write $\omega_\alpha^{\otimes n}=\sum_{\mbf{i},\mbf{j}}r_{\mbf{i}\mbf{j}}\op{\mbf{i}}{\mbf{j}}$, where $\op{\mbf{i}}{\mbf{j}}=\op{i_0}{j_0}\otimes\op{i_1}{j_1}\otimes\cdots\otimes\op{i_{n-1}}{j_{n-1}}$.  Since $[[\omega_\alpha]]_{i,j}=1/2$ if $i\oplus j=0$ and $[[\omega_\alpha]]_{i,j}=\alpha/2$ if $i\oplus j=1$, we can express $\omega_\alpha^{\otimes n}$ as
\begin{equation}
\omega_\alpha^{\otimes n}=\frac{1}{2^n}\sum_{\mbf{i},\mbf{j}\in\{0,1\}^{\times n}}\alpha^{n_{\mbf{i}\mbf{j}}}\op{\mbf{i}}{\mbf{j}},
\end{equation}
where $n_{\mbf{i}\mbf{j}}=\sum_{k=0}^{n-1}(i_k\oplus j_k)$.  Our goal now is to further express $\omega_\alpha^{\otimes n}$ as a convex combination of pure states, each having coherence number no larger than two.  We claim that
\[\omega_\alpha^{\otimes n}=\frac{1}{2^n}\sum_{\mbf{i}\not=\mbf{j}}\alpha^{n_{\mbf{i}\mbf{j}}}\op{A_{\mbf{i}\mbf{j}}}{A_{\mbf{i}\mbf{j}}}+\frac{2-(1+\alpha)^n}{2^n}\mbb{I},\]
where
\[\op{A_{\mbf{i}\mbf{j}}}{A_{\mbf{i}\mbf{j}}}=\frac{1}{2}\left(\op{\mbf{i}}{\mbf{i}}+\op{\mbf{i}}{\mbf{j}}+\op{\mbf{j}}{\mbf{i}}+\op{\mbf{j}}{\mbf{j}}\right).\]
Clearly this decomposition correctly recovers all the off-diagonal elements of $\omega_\alpha^{\otimes n}$.  To check correctness of the diagonal terms, first notice that by symmetry
\[\sum_{\mbf{j}\in\{0,1\}^{\times n}}\alpha^{n_{\mbf{i}\mbf{j}}}=\sum_{\mbf{j}\in\{0,1\}^{\times n}}\alpha^{n_{\mbf{i}'\mbf{j}}}\]
for any sequences $\mbf{i}$ and $\mbf{i}'$, which implies that every diagonal term is the same in the part $\frac{1}{2^n}\sum_{\mbf{i}\not=\mbf{j}}\alpha^{n_{\mbf{i}\mbf{j}}}\op{A_{\mbf{i}\mbf{j}}}{A_{\mbf{i}\mbf{j}}}$.  Hence it suffices to consider the coefficient of $\op{0}{0}$, which is given by $\frac{1}{2^n}\sum_{\mbf{j}\not=0}\alpha^{n_{0\mbf{j}}}$.  The number $n_{0\mbf{j}}$ is simply the number of ones in $\mbf{j}$, and so
\begin{align}
\sum_{\mbf{j}\not=0}\alpha^{n_{0\mbf{j}}}=\sum_{k=1}^{n}\alpha^k\binom{n}{k}=(1+\alpha)^n -1\leq 1,
\end{align}
by the assumption $\alpha\leq 2^{1/n}-1$.  Adding the identity by an appropriate amount thus recovers the correct diagonal terms of value $\frac{1}{2^n}$.
\end{proof}

Note that Theorem \ref{thm:1} only guarantees strong non-additivity for $\alpha\leq 2^{1/n}-1$.  In contrast, when $\alpha=1$, $\omega_\alpha$ is pure and $r_C(\omega_\alpha^{\otimes n})=2^n$.  Hence, the coherence rank of $\omega^{\otimes n}_\alpha$ varies between $2$ and $2^n$ for $2^{1/n}-1\leq \alpha\leq 1$.  For $\alpha>\sqrt{2}-1$, we leave it as an open problem to find the coherence rank of $\omega^{\otimes n}_\alpha$.  

The next task is to consider the regularized zero-error coherence cost of $\omega_\alpha$.  We first observe a lower bound given in terms of the $\ell_1$ coherence norm \cite{Baumgratz-2014a}, which is defined as
\begin{equation}
||\rho||_{\ell_1}=\sum_{i\not=j}|\bra{i}\rho\ket{j}|.
\end{equation}
In other words, the $\ell_1$ norm sums the magnitude of all off-diagonal terms of a given density matrix.  Up to a permutation of incoherent basis states, we can express any pure state as
\begin{equation}
\ket{\psi}=\sum_{i=1}^{r_c(\psi)}c_i\ket{i}.
\end{equation}
It is well-known that the transformation $\ket{\Psi_{r_c(\psi)}}\to\ket{\psi}$ is always possible using IO \cite{Du-2015a, Winter-2016b}.  Then since the $\ell_1$ norm is a monotone under IO, and since $||\phi_{m}||_{\ell_1}=m-1$, it follows that
\begin{equation}
\label{Eq:l1-bound}
r_c(\psi)\geq ||\psi||_{\ell_1}+1,
\end{equation}
with equality holding whenever $\ket{\psi}$ is maximally coherent. Consider now a mixed state $\rho$ such that $r_c(\rho)=m$.  There must exist a pure state ensemble $\{p_k,\ket{\varphi_k}\}$ for $\rho$ such that $r_c(\varphi_k)\leq m$.  Using the previous inequality we therefore obtain
\begin{align}
\label{Eq:l1-bound-general}
r_c(\rho)&\geq \sum_{k}p_kr_c(\varphi_k)\geq \sum_{k}p_k||\varphi_k||_{\ell_1}+1\geq ||\rho||_{\ell_1}+1,
\end{align}
where the last line follows from convexity of the $\ell_1$ norm.  Equation \eqref{Eq:l1-bound-general} is a general bound for all states, and when applied to $\omega^{\otimes n}_\alpha$ we obtain
\begin{align}
\label{Eq:Coherenc_Rank_lower_bound}
r_c(\omega^{\otimes n}_\alpha)&\geq ||\omega^{\otimes n}_\alpha||_{\ell_1}+1\notag\\
&=\frac{1}{2^n}\sum_{\mbf{i}\not=\mbf{j}}\alpha^{n_{\mbf{i}\mbf{j}}}+1=(1+\alpha)^n.
\end{align}
Combining this bound with Theorem \ref{thm:1} yields the following.
\begin{theorem}
\label{thm:2}
For $0<\alpha\leq \sqrt{2}-1$,
\begin{equation}
\left\lfloor\frac{1}{\log_2(\alpha+1)}\right\rfloor^{-1} \geq C_{c,IO}^{(0,\infty)}(\omega_\alpha)\geq \log_2(\alpha+1),
\end{equation}
where $\lfloor x\rfloor$ denotes the largest integer no greater than $x$.
\end{theorem}
\begin{proof}
The lower bound follows immediately from Eq. \eqref{Eq:Coherenc_Rank_lower_bound} and the fact that the coherence rank captures the zero-error coherence cost under IO.  For the upper bound, take $m=\left\lfloor\frac{1}{\log(\alpha+1)}\right\rfloor$.  Then Theorem \ref{thm:1} implies that $r_c(\omega_\alpha^{\otimes m})=2$. Hence
\begin{align}
\lim_{n\to\infty}\frac{1}{n}\log_2[r_c(\omega_{\alpha}^{\otimes n})]&=\lim_{n\to\infty}\frac{1}{nm} \log_2[r_c(\omega_\alpha^{\otimes nm})]\notag\\
&\leq \frac{1}{m}\log_2[r_c(\omega_\alpha^{\otimes m})]=\frac{1}{m},
\end{align}
where we use the inequality $r_c(\omega_\alpha^{\otimes nm})\leq r_c(\omega_\alpha^{\otimes m})^n$.
\end{proof}
Whenever $\alpha=2^{1/n}-1$ for some integer $n$, the upper and lower bounds coincide and
\begin{equation}
C_{c,IO}^{(0,\infty)}(\omega_\alpha)= \log_2(\alpha+1).
\end{equation}

\section{The Coherence Number Under Dephasing-Covariant Operations}

\label{Sect:DIO}

We next consider the question of increasing the Schmidt rank by quantum operations.  As described above, the zero-error coherence cost under IO and SIO is given precisely by the logarithm of the coherence rank.  Hence, the coherence rank must be a monotone under IO.  However, under DIO, the zero coherence cost is given by a different quantity, as shown in Ref. \cite{Zhao-2018a} and reviewed below.  It then remains a question of whether the coherence rank is still a monotone under DIO.  Here we find that DIO is indeed a monotone for pure state transformations (see also \cite{Chitambar-2016b}), whereas for mixed states it can be increased by an arbitrarily large amount.

For an arbitrary density matrix $\rho$, consider the transformation $\phi_d\to\rho$ via some CPTP map $\Lambda$.  To obtain necessary and sufficient conditions for such a map to be DIO, let $\Omega:=\id\otimes\Lambda\left(\sum_{i,j=1}^d\op{ii}{jj}\right)$ denote its Choi matrix \cite{Choi-1975a}.  Since $\phi_d$ is invariant under permutations of basis states (a valid DIO operation), we can symmetrize over all permutations applied to the first system in $\Omega$, and without loss of generality we have that
\begin{equation}
\Omega=\phi_d\otimes A+(\mbb{I}_d-\phi_d)\otimes B,
\end{equation}
where $A,B\geq 0$ and $\tr[A]=\tr[B]=1$.  For any state $\sigma$, we compute
\begin{align}
\Lambda(\Delta(\sigma^T))&=\frac{1}{d}A+(1-\frac{1}{d})B\notag\\
\Delta(\Lambda(\sigma))&=\bra{\phi_d}\sigma\ket{\phi_d}\Delta(A)+(1-\bra{\phi_d}\sigma\ket{\phi_d})\Delta(B).\notag
\end{align}
For $\Lambda$ to be DIO, these lines must be equivalent for all $\sigma$, meaning that
\begin{equation}
\label{Eq:DIO-constraint}
\frac{1}{d}A+(1-\frac{1}{d})B=\Delta(B)=\Delta(A).
\end{equation}
In summary, we have $\phi_d\to\rho=:A$ by DIO if and only if there exists a diagonal matrix $D$ with $\tr[D]=1$ and off-diagonal matrix $Z$ such that
\begin{align}
A&=D+Z\geq 0\\
B&=D-\frac{1}{d-1}Z\geq 0.
\end{align}
Solving for $Z$ in terms of $A$, we reach the following.
\begin{proposition}[\cite{Zhao-2018a}]
\label{Prop:DIO-MaxCoherent}
The transformation $\phi_d\to \rho$ is possible by DIO if and only if
\begin{equation}
\label{Eq:Delta-robustness}
d\Delta(\rho)-\rho\geq 0.
\end{equation}
\end{proposition}
In Ref. \cite{Chitambar-2016a}, the minimum $d$ that satisfies Eq. \eqref{Eq:Delta-robustness} was identified as the $\Delta$-robustness of coherence of $\rho$.  This quantity is easily seen to be monotonically decreasing under DIO, and for pure states, it coincides with the coherence rank \cite{Chitambar-2016b}.
For completeness, we here give a quick proof of this fact.
\begin{corollary}[\cite{Chitambar-2016b}]
The coherence rank is monotonically decreasing under DIO for pure state transformations; i.e. for transformations of the form $\ket{\psi}\to\ket{\tau}$.
\end{corollary}
\begin{proof}
If $\ket{\psi}\to\ket{\tau}$ is possible by DIO then so is $\ket{\phi_{r_c(\psi)}}\to\ket{\tau}$.  Proposition \ref{Prop:DIO-MaxCoherent} then says that $r_c(\psi)\cdot\Delta(\tau)\geq\tau$.  Writing $\ket{\tau}=\sum_{i=1}^{r_c(\tau)}c_i\ket{i}$, we thus have
\begin{equation}
r_c(\psi)\geq 1/\sum_{i=1}^{r_c(\tau)}|c_i|^4\geq r_c(\tau).
\end{equation}
\end{proof}
We next show by an explicit example that monotonicity under DIO fails dramatically when the target state in the transformation is mixed.  For each non-negative integer $d$, consider the family of states $\{\ket{\varphi_k}\}_{k=0}^{d-1}$ belonging to $\mbb{C}^{2d}$, given by
\begin{align}
\ket{\varphi_{k}}:=\frac{1}{\sqrt{d+1}}\left(\ket{k}+\sum_{j=0}^{d-1}e^{-ijk(2\pi)/d}\ket{d+j}\right).
\end{align}
Each $\ket{\varphi_k}$ can be interpreted as the state obtained by starting with the initial superposition
\[\frac{1}{\sqrt{d+1}}\ket{k}+\frac{\sqrt{d}}{\sqrt{d+1}}\ket{d+k}\]
and performing a $d$-dimensional Fourier transform on the basis vectors $\{\ket{d+k}\}_{k=0}^{d-1}$.  To make this more transparent, we write
\begin{equation}
\ket{\varphi_k}=\frac{1}{\sqrt{d+1}}\ket{k}+\frac{\sqrt{d}}{\sqrt{d+1}}\ket{\widetilde{d+k}},
\end{equation}
where
\[\ket{\widetilde{d+k}}:=\frac{1}{\sqrt{d}}\sum_{j=0}^{d-1}e^{-ijk(2\pi)/d}\ket{d+j}.\]
A crucial property of the $\ket{\varphi_k}$ for our purposes is the following.
\begin{lemma}
\label{Lem:Coherence-number}
Any linear combination $\ket{\kappa}=\sum_{k=0}^{d-1}a_k\ket{\varphi_k}$ has coherence rank at least $d+1$.
\end{lemma}
\begin{proof}
We expand
\begin{equation}
\ket{\kappa}=\frac{1}{\sqrt{d+1}}\left(\sum_{k=0}^{d-1}a_k\ket{k}+\sum_{j,k=0}^{d-1}a_ke^{-ijk(2\pi)/d}\ket{d+j}\right).\notag
\end{equation}
Let us suppose, without loss of generality, that all the $a_k$ are zero except for the first $t$ elements, $\{a_0,a_1,\cdots,a_t\}$.  Then it is easy to see that
\begin{equation}
r_c(\ket{\kappa})=t+r_c\left(\sum_{k=0}^t\sum_{j=0}^{d-1} a_ke^{-ijk(2\pi)/d}\ket{d+j}\right).
\end{equation}
Observe that the coefficients of the vector in parenthesis are given by the column matrix
\[\vec{v}=\sqrt{d}U^{(t)}_{FT}\vec{a},\]
where $\vec{a}=[a_1,a_2,\cdots,a_t]^T$ and $U^{(t)}_{FT}$ is the first $t$ columns of the Fourier transform matrix, i.e. $[[U^{(t)}_{FT}]]_{j,k}=\frac{1}{\sqrt{d}}e^{-ijk(2\pi)/d}$.  Since $U_{FT}$ is invertible, the columns of $U^{(t)}_{FT}$ are linearly independent, and it follows that $\vec{v}$ cannot have more than $t-1$ vanishing elements.  Therefore $r_c\left(\sum_{k=0}^t\sum_{j=0}^{d-1} a_ke^{-ijk(2\pi)/d}\ket{d+j}\right)\geq d-(t-1)$ and so $r_c(\ket{\kappa})\geq d+1$.
\end{proof}
We next consider the uniform mixture
\begin{equation}
\rho_d:=\frac{1}{d}\sum_{k=0}^{d-1}\op{\varphi_k}{\varphi_k}.
\end{equation}
Its relevance to the study of coherence rank is given in the following.
\begin{theorem}
\label{thm:3}
For any $d\in\mbb{N}$, $r_c(\rho_d)=d+1$, and moreover the rank-increasing transformation $\phi_2\to\rho_d$ is possible by DIO
\end{theorem}
\begin{proof}
From Lemma \ref{Lem:Coherence-number} it follows that $r_c(\rho_d)\geq d+1$ since any pure-state decomposition of $\rho_d$ must involve linear combinations of the $\ket{\varphi_k}$.  To show that this lower bound is tight,  we explicitly construct a pure state ensemble for $\rho_d$, each state having coherence number of exactly $d+1$.  Letting $u_{j,k}=\frac{1}{\sqrt{d}}e^{ijk(2\pi)/d}$ be the components of the inverse Fourier transform, it is straightforward to compute that
\begin{align}
\ket{\varphi_j'}:&=\sum_{k=0}^{d-1}u_{j,k}\ket{\varphi_k}\notag\\
&=\frac{1}{\sqrt{d(d+1)}}\left(\sum_{k=0}^{d-1}e^{ijk(2\pi)/d}\ket{k}+d\ket{d+j}\right).\notag
\end{align}
Hence $\rho_d=\frac{1}{d}\sum_{j=0}^{d-1}\op{\varphi_j'}{\varphi_j'}$ and so $r_c(\rho_d)=d+1$.  Turning to the convertibility under DIO, we write
\[\rho_d=\Delta(\rho_d)+Z\]
and observe that $Z$ lies in the linear span of the matrices $\{\op{k}{d+j}\}_{j,k=0}^{d-1}$.  Let $U_Z$ be the diagonal matrix that multiplies each $\ket{d+j}$ by an overall $-1$ phase and acts trivially on basis vectors $\ket{k}$; that is,
\begin{align}
U_Z\ket{\varphi_k}&=\frac{1}{\sqrt{d+1}}\left(\ket{k}-\sum_{j=0}^{d-1}e^{-ijk(2\pi)/d}\ket{d+j}\right)\notag\\
&=\frac{1}{\sqrt{d+1}}\ket{k}-\frac{\sqrt{d}}{\sqrt{d+1}}\ket{\widetilde{d+k}}.
\end{align}
As a result, we have
\begin{equation}
U_Z(\rho_d) U_Z^\dagger=\Delta(\rho_d)-Z\geq 0.
\end{equation}
By Proposition \ref{Prop:DIO-MaxCoherent}, the transformation $\Psi_2\to\rho_d$ can be accomplished by a DIO map.
\end{proof}

\section{Non-Additivity of the Schmidt Number and Maximally-Correlated Dephasing-Covariant Operations}

\label{Sect:MCDC}

The final task of this paper is to convert the previous results into analogous statements about entangled bipartite states.  Such a mapping between coherence and entanglement has been exploited previously in different contexts.  The basic observation is that every $d$-dimensional density matrix $\rho$ can be associated with a $d\otimes d$ \textit{maximally correlated} (MC) state $\wh{\rho}$ according to
\begin{equation}
\label{Eq:coherence-MC}
\rho=\sum_{i,j=1}^dc_{ij}\op{i}{j}\;\Leftrightarrow\;\wh{\rho}=\sum_{i,j=1}^d c_{ij}\op{ii}{jj}.
\end{equation}
Measures of coherence in $\rho$ translate into measures of entanglement in $\wh{\rho}$ \cite{Streltsov-2015a, Zhu-2017a}; in particular
\begin{equation}
r_c(\rho)=r_S(\wh{\rho}).
\end{equation}
With this correspondence, Theorems \ref{thm:1} and \ref{thm:2} imply that
\begin{align}
r_S(\wh{\omega}^{\otimes n}_\alpha)=2
\end{align}
for $0<\alpha\leq 2^{1/n}-1$, and
\begin{equation}
\left\lfloor\frac{1}{\log_2(\alpha+1)}\right\rfloor^{-1} \geq E_{c}^{(0,\infty)}(\wh{\omega}_\alpha)\geq \log_2(\alpha+1),
\end{equation}
for $0<\alpha<\sqrt{2}-1$, where
\begin{equation}
\wh{\omega}_{\alpha}=\frac{1-\alpha}{2}(\op{00}{00}+\op{11}{11})+\alpha\Phi_2.
\end{equation}
Interestingly, the Entanglement of Formation is additive for $\wh{\omega}_{\alpha}$ \cite{Vidal-2002c, Horodecki-2003b}, and thus the entanglement cost can be easily computed \cite{Wootters-1998a} as
\begin{equation}
E_c(\wh{\omega}_\alpha)=h\left(\frac{1}{2}(1-\sqrt{1-\alpha^2})\right),
\end{equation}
where $h(x)=-x\log_2 x-(1-x)\log_2(1-x)$.  Thus, whenever $\log_2(\alpha+1)$ is an integer, $\wh{\omega}_\alpha$ provides a rare example in which the zero-error entanglement cost, the regularized zero-error entanglement cost, and the asymptotic entanglement cost can all be computed.

As a remark, we note that the lower bound of Eq. \eqref{Eq:l1-bound} translates into
\begin{equation}
r_S(\wh{\rho})\geq 2\mc{N}(\wh{\rho})+1,
\end{equation}
where $\mc{N}(\wh{\rho})=\frac{1}{2}(||\wh{\rho}^\Gamma||-1)$ is the negativity of $\wh{\rho}$ and $\wh{\rho}^\Gamma$ denotes its partial transpose \cite{Vidal-2002b}.  This is due to the equivalence $||\rho||_{\ell_1}=||\wh{\rho}^\Gamma||_1$ \cite{Zhu-2017a, Chitambar-2018c}.  In fact, however, this bound has already been established by Eltschka \textit{et al.} using direct calculation \cite{Eltschka-2015a}.

A final question is whether there exists any statement similar to Theorem \ref{thm:3} in entanglement theory.  More precisely, does there exist a class of bipartite operations that cannot increase the pure-state Schmidt rank but can increase the mixed-state Schmidt rank by an arbitrarily large amount?  Such an operational class indeed exists, and it represents an extension of DIO to bipartite systems.  These have been called maximally-correlated dephasing-covariant (MCDC) operations \cite{Chitambar-2018c}, and it consists of local unitaries along with any map of the form $\wh{\mc{E}}=\tau\circ\mc{E}\circ\tau$, in which $\tau(\rho)=\int_{\mc{G}} dg (U_g\otimes U^*_g)\rho(U_g\otimes U^*_g)^\dagger$ is a twirling map with $\mc{G}=\{U_g\}_g$ being the group of unitary matrices diagonal in the incoherent basis $\{\ket{i}\}_{i=0}^{d-1}$, and $\wh{\mc{E}}$ is any CPTP map that commutes with the maximally-correlated dephasing map $\wh{\Delta}(\cdot)=\sum_{i=0}^{d-1}\op{ii}{ii}(\cdot)\op{ii}{ii}$ while acting invariantly on all uncorrelated states $\op{ij}{ij}$, $i\not=j$.  With this form, an operational correspondence is established
\begin{align}
\label{Eq:DIO-MCDC}
\Lambda(\rho)=\sigma\;\;\Leftrightarrow\;\;\wh{\Lambda}(\wh{\rho})=\wh{\sigma},
\end{align}
where $\Lambda\in \text{DIO}$ and $\wh{\Lambda}\in\text{MCDC}$.  Consequently, we have the following corollary to Theorem \ref{thm:3}.
\begin{corollary}
The Schmidt rank is monotonically decreasing under the MCDC transformation of pure states, yet it can be increased by an arbitrarily large amount when acting on mixed states.
\end{corollary}
What makes this finding particularly noteworthy is that MCDC operations are no weaker than LOCC for pure state convertibility.  That is, if $\ket{\psi}\xrightarrow{\text{LOCC}}\ket{\phi}$, then $\ket{\psi}\xrightarrow{\text{MCDC}}\ket{\phi}$.  This follows from the facts that (i) local unitaries belong to the class of MCDC (so every pure state can be put in the form $\ket{\wh{\psi}}=\sum_{i=0}^{d-1}c_i\ket{ii}$), (ii) Eq. \eqref{Eq:DIO-MCDC}, (iii) the majorization of Schmidt coefficients is necessary and sufficient for LOCC convertibility of pure states, and (iv) the majorization of coherent amplitudes is sufficient for DIO convertibility of pure states.  It remains an open question of whether the power of MCDC is actually equivalent to LOCC for pure state transformations.  This question is equivalent to whether the majorization condition is also necessary for DIO convertibility.

\section{Conclusion}

This paper has investigated certain behavior of the Schmidt number using tools of quantum coherence theory.  In particular, the Schmidt number has been found to be highly non-additive, even for simple two-qubit states.  In bipartite systems, the Schmidt number provides the unique measure of stochastic LOCC (SLOCC) entanglement classification, meaning that two states $\ket{\psi}$ and $\ket{\phi}$ can be reversibly transformed from one to the other using SLOCC if and only if they have the same Schmidt number \cite{Dur-2000a}.  We have identified a class of operations (MCDC) that preserves this SLOCC entanglement classification under pure state transformations.  Nevertheless, when converting mixed states, these operations can increase Schmidt number by large amounts.  This demonstrates that the Schmidt number behaves quite differently in pure states versus mixed states, particularly when transforming states beyond the LOCC framework.

A central focus in this paper has been computing the regularized zero-error coherence cost of $\omega_\alpha$. For the situation of $\alpha\leq\sqrt{2}-1$, we have obtained near-tight lower and upper bounds.  We suspect that similar results can be obtained by performing an analogous combinatorial analysis on the multi-copy level.  We leave this for future work.  The techniques in this paper can be applied to the $d$-dimensional generalization of $\omega_\alpha$, and strong non-additivity can likewise be observed.  This can be used to compute lower bounds on the Schmidt number of maximally-correlated states of any size.



\medskip

\noindent\textbf{Acknowledgments:}  After the presentation of these results at the Seefeld quantum information workshop, we became aware of independent work by Bartosz Regula, Ludovico Lami, and Gerardo Adesso that obtains Theorems 1 and 2 using different techniques.  Qiuling Yue is supported by China Scholarship Council.  This work is also supported by the National Science Foundation (NSF) Early CAREER Award No. 1352326.

\bibliography{Schmidt_Number}
\end{document}